\documentclass[12pt,twoside]{article}

\usepackage[a4paper,top=2.0cm,bottom=2.0cm,left=2.5cm,right=2.5cm,nofoot]{geometry}
\usepackage[english]{babel}     
\usepackage[T1]{fontenc}       
\usepackage[ansinew]{inputenc} 

\usepackage{amsmath}
\usepackage{amsthm}
\usepackage{amssymb}
\usepackage{url}
\usepackage[numbers]{natbib}
\usepackage{graphicx,dblfloatfix}
\usepackage{float}
\usepackage[nofillcomment,vlined,linesnumbered,titlenumbered,algo2e,ruled]{algorithm2e}

\usepackage[noblocks]{authblk}
\usepackage[center,small]{titlesec}

\newcommand{\runningheads}{\markboth}
\theoremstyle{plain}
\newtheorem{theorem}{Theorem}[section]
\newtheorem{lemma}[theorem]{Lemma}
\newtheorem{proposition}[theorem]{Proposition}

\newtheorem{example}[theorem]{Example}

\title{\vspace{3.5cm}\bf ALGORITHMS AND PROPERTIES FOR POSITIVE SYMMETRIZABLE MATRICES}

\author[1]{Elis\^angela Silva Dias}
\author[1]{Diane Castonguay}
\author[2]{Mitre Costa Dourado}

\affil[1]{Instituto de Informática, Universidade Federal de Goi\'as -- Alameda Palmeiras, Quadra D, Câmpus Samambaia, Caixa Postal 131 - CEP 74001-970 -- Goiânia -- GO, BRAZIL

e-mail: \{elisangela,diane\}@inf.ufg.br}

\affil[2]{Departamento de Ci\^{e}ncia da Computa\c{c}\~{a}o -- Instituto de Matem\'{a}tica, Universidade Federal do Rio de Janeiro -- Av. Athos da Silveira Ramos, 274, Ilha do Fundão -- CEP 21941-916 -- Rio de Janeiro -- RJ, BRAZIL

e-mail: mitre@dcc.ufrj.br}

\begin{document}

\maketitle

\runningheads{E.S. Dias, D. Castonguay, M.C. Dourado}{Algorithms and Properties for Positive Symmetrizable Matrices}

\noindent {\bf Abstract:} Matrices are the most common representations of graphs. They are also used for the representation of algebras and cluster algebras. This paper shows some properties of matrices in order to facilitate the understanding and locating symmetrizable matrices with specific characteristics, called positive quasi-Cartan companion matrices. Here, symmetrizable matrix are those which are symmetric when multiplied by a diagonal matrix with positive entries called symmetrizer matrix. Four algorithms are developed: one to decide whether there is a symmetrizer matrix; second to find such symmetrizer matrix; another to decide whether the matrix is positive or not; and the last to find a positive quasi-Cartan companion matrix, if there exists. The third algorithm is used to prove that the problem to decide if a matrix has a positive quasi-Cartan companion is NP.\\[-10 pt]

\noindent {\bf Math. Subj. Classification 2010:} 05B20, 13F60, 15A15.

\noindent {\bf Key Words:} Symmetrizable matrix, Positive quasi-Cartan matrix, Algorithm.

\section{Introduction}
\label{sec:intro}

The matrices can be used to represent various structures, including
graphs and algebras, such as cluster algebra. It can be defined using a directed graph $G(B)$, called quiver, and consequently by an adjcency matrix, where rows and columns represent the vertices and the positive values at positions $(i, j)$ represent the quantity of edges between associated vertices of the graph. For more information about quiver and cluster algebras, see~\cite{BGZ2006,FZ2002,FZ2003}.

Cartan matrices were introduced by the French mathematician \'{E}lie Cartan. In fact, Cartan matrices, in the context of Lie algebras, were first investigated by Wilhelm Killing, whereas the Killing form is due to Cartan.

The notion of quasi-Cartan matrices was introduced by Barot, Geiss and Zelevinsky~\cite{BGZ2006}. They show some properties of the matrices, of the mathematical point of view. Those matrices are symmetrizable, called quasi-Cartan companion, associated to skew-symmetrizable matrices.

By Sylvester criterion~\cite{W1976}, a symmetric matrix is positive if
all leading principal submatrices have positive determinant. We see
that a symmetrizable matrix is positive if its associated symmetric
matrix also is. One can decide whether a cluster algebra is of finite
type (has a finite number of cluster variables) deciding whether it
has a quasi-Cartan companion matrix which is positive, along with
other criteria that will not be discussed in this paper.

In this paper, we study the matrices of mathematical and computational point of view and we verify some inherent properties of the positive quasi-Cartan matrix. We also developed four algorithms. The first one decides in time complexity $\theta(n ^ 2)$ if the matrix is symmetrizable and returns the symmetrizer if it exists. The second one  find a symmetrizer matrix for a symmetrizable matrix, having time complexity $\theta(n^2)$ in the worst case and $\theta(n)$ in the best case. The next decides whether the matrix is positive or not with time complexity $\theta(n^4)$. It is used as prove that the problem to decide if a matrix has a positive quasi-Cartan companion is in NP class. The last algorithm is exponential and it finds a positive quasi-Cartan companion matrix, for a skew-symmetrizable matrix, if there exists.

\section{Preliminaries}
\label{sec:preliminaries}

In this paper, we considered square matrices with integer entries, except the matrix $D$. Let $n$ be a positive integer, $A, B, C \in \mathrm{M}_n(\mathbb {Z})$ and $D \in \mathrm{M}_n(\mathbb {R})$. A matrix $A$ is \textit {symmetric} if $A = A^T$, where $A^T$ is the transpose of $A$. A matrix $C$ is \textit{symmetric by signs} if for all $i, j \in \{1, \dots, n\}$, with $i \neq j$, we have $c_{ij} = c_{ji} = 0$ or $c_{ij} \cdot c_{ji} > 0$. A matrix $C$ is \textit {symmetrizable} if $D \times C$ is symmetric for some diagonal matrix $D$ with positive diagonal entries. In this case, the matrix $D \times C$ is called \textit {symmetrization or symmetrized} of $C$ and the matrix $D$ is called \textit {symmetrizer} of $C$. Note that this definition is equivalent that the one given in~\cite{CM1980}.

The matrix $A$ is \textit{skew-symmetric} if its transpose coincides with its opposite ($A^T = -A$), i.e., $a_{ij} = -a_{ji}$, for all $i, j$. Observe that the values of the main diagonal are null. A matrix $B$ is \textit{skew-symmetric by signs} if for all $i, j \in \{1, \dots, n\}$ we have $b_{ii} = 0$ and if $i \neq j$, then $b_{ij} = b_{ji} = 0$ or $b_{ij} \cdot b_{ji} < 0$. The matrix $B$ is \textit{skew-symmetrizable} if there exists a diagonal matrix $D$ with positive entries such that $D \times B$ is a skew-symmetric matrix. In this case, the matrix $D \times B$ is called \textit{skew-symmetrization} or {\it skew-symmetrized} of $B$ and the matrix $D$ of \textit{skew-symmetrizer} of $B$.

We must observe that every symmetric matrix is symmetrizable and that every skew-symmetric matrix is skew-symmetrizable. Also observe that all symmetrizable matrices are symmetric by signs and all skew-symmetrizable matrices are skew-symmetric by signs.

A symmetrizable matrix is \textit{quasi-Cartan} if all entries of main diagonal are equal to 2. For a skew-symmetrizable matrix $B$, we will refer to a quasi-Cartan matrix $C$ with $|c_{ij}| = |b_{ij}|$ for all $i \neq j$ as a \textit {quasi-Cartan companion} of $B$.

Given a skew-symmetrizable matrix $B$, we want to find a positive quasi-Cartan companion of $B$. For this, we need one more definition.

The matrix $A_{[ij]}$ is obtained by elimination of $i^{th}$ row and $j^{th}$ column of matrix $A$. The $ij^{th}$ \textit{minor} of $A$ is the determinant of $A_{[ij]}$. Recall that the determinant of $A$ can be recursively defined in terms of their minors. For more information, see~\cite{L2003}.

A {\it principal submatrix} of $A$ is a submatrix of $A$ obtained by eliminating some rows and respective columns of $A$.

The \textit{principal minors} are the determinants of all principal submatrices of $A$.

The \textit{leading principal minors} are the determinants of the diagonal blocks of a matrix $M$ with dimension $1, 2, \ldots, n$. These submatrices, also called \textit{leading principal matrices}, are obtained by eliminating the last $k$ columns and $k$ rows, with $k = n-1, n-2, \ldots, 0$.

For any arbitrary matrices $A$ of dimension $m \times n$ and $B$ of dimension $p \times q$, we define the {\it direct sum} of $A$ and $B$, denoted by
$A \oplus B =
  \begin{bmatrix}
     A & 0 \\
     0 & B \\
  \end{bmatrix}$.

Note that any element in the direct sum of two vector spaces of matrices could be represented as a direct sum of two matrices.

A {\it permutation matrix} is a square binary matrix that has exactly one entry 1 in each row and each column and 0s elsewhere. Each such matrix represents a specific permutation of $n$ elements and, when used to multiply another matrix, can produce that permutation in the rows or columns of the other matrix.

Let $A$ be a square matrix. We say that $A$ is {\it disconnected} if there exists $P$ permutation matrices such that $PAP$ is a direct sum  of at least two non-zero matrix. If not, we say $A$ is {\it connected}. Observe that $PAP$ is obtained from $A$ by permutation of rows and respective columns. Moreover, $A$ is connected exaclty when the graph associated to the incidence matrix $A$ is.


\section{Symmetrizable and skew-symmetrizable matrices}
\label{sec:symmetrizer_matrix}

In this section, we show some properties of symmetrizable and skew-symmetrizable matrices. We also present two algorithms for symmetrizable matrix.


%
%
%

The following proposition help us to find in an easy way a quasi-Cartan companion for a skew-symmetrizable matrix.

\begin{theorem}\label{prop:same_symmetrizer}
Let $B$ be a skew-symmetrizable matrix. Consider a matrix $C$ such that $|c_{ij}| = |b_{ij}|$, for all $i \neq j$. If $C$ is symmetric by the signs, then $C$ is symmetrizable with the same symmetrizer of $B$. Furthermore, if $c_{ii} = 2$, for all $i$, then $C$ is a quasi-Cartan companion of $B$.
\end{theorem}

\begin{proof}
Let $D =  \left(
        \begin{array}{ccc}
          d_1 & & \\
          & \vdots & \\
          & & d_n \\
        \end{array}
      \right)$
be a symmetrizer matrix of $B$. Then, we have that $d_i \cdot |c_{ij}| = d_i \cdot |b_{ij}| = |d_i \cdot b_{ij}| = |d_j \cdot b_{ji}| = d_j \cdot |b_{ji}| = d_j \cdot |c_{ji}|$. Since $c_{ij}$ and $c_{ji}$ have the same sign, we have that $d_i \cdot c_{ij} = d_j \cdot c_{ji}$.
\end{proof}

The Lemma~\ref{lem:symmetrizable} follows from \cite{FZ2003}(Lemma 7.4).

\begin{lemma}  \label{lem:symmetrizable}
A matrix $C$ is symmetrizable if and only if it is  symmetric by signs and for all $k \geq 3$ and all $i_1, i_2, \ldots, i_k$ it satisfies:
$$ c_{i_1 i_2} \cdot c_{i_2 i_3} \cdot \ldots \cdot c_{i_k i_1} = c_{i_2 i_1} \cdot c_{i_3 i_2} \cdot \ldots \cdot c_{i_1 i_k}.$$
\end{lemma}

We present two algorithms, one for deciding whether any matrix has a symmetrizer and returns it if it exists, and another to find symmetrizer matrix for a symmetrizable matrix. The second algorithm has time complexity $\theta(n^2)$ in the worst case and $\theta(n)$ in the best case, and the first algorithm has time complexity $\theta(n^2)$. Thus, if we know the information that a given matrix is symmetrizable, is more suitable using second algorithm.

\begin{center}
\begin{algorithm2e}[H]
 \BlankLine
 \KwIn{A $n \times n$ matrix $A$.}
 \KwOut{If the matrix is symmetrizable or not; and a diagonal matrix $D = d_{ii}$ of positive values, if there exists, such that $D \times A$ is symmetric.}
 \BlankLine

\ForEach{$i \in \{1, \ldots, n\}$}{ \label{sym-1}
	$d_{ii} \leftarrow 0$\label{sym-2}
}

$T \leftarrow \{1, \dots, n\}$ \label{sym-4} ordered list

\While{$T \neq \varnothing$}{ \label{sym-5}

	 $i \leftarrow $ the first element of $T$ \\ \label{sym-6}
    	$T \leftarrow T \setminus \{i\}$\\ \label{sym-7}
 \If{$d_{ii} = 0$}{
					  $d_{ii} \leftarrow 1$    }

	  \ForEach{$j \in T$} { \label{sym-9}

			   \eIf{$a_{ij} \cdot a_{ji} = 0$}{
				      \If{$a_{ij} + a_{ji} \neq 0$}{
					         \Return NO
				      }
			   }
         {

   				         move $j$ for the first position of $T$

				           \eIf{$d_{jj} \neq 0$}{
					              \If{$d_{ii} \cdot a_{ij} \neq d_{jj} \cdot a_{ji}$}{
						                 \Return NO
					              }
                   }
                   {
    			              $d_{jj} \leftarrow \frac{d_{ii} \cdot a_{ij}}{a_{ji}}$
					         }
				      }
	       }
    }

\Return YES, $D$ \label{sym-21}

\BlankLine

\BlankLine

\caption{$SymmetrizableMatrix(A)$ \label{alg:symmetrizable_matrix}}
\end{algorithm2e}
\end{center}

\pagebreak
Observe that if $A$ is connected then we effectuate the line 8 only once.
 
\begin{proposition}
Algorithm~$\ref{alg:symmetrizable_matrix}$ is correct.
\end{proposition}

\begin{proof}
At the beginning of any iteration of the ``while loop'', $d_{ii} \cdot a_{ij} = d_{jj} \cdot a_{ji}$ for any $j \notin T$ and any $i$.
This is clearly true since at the beginning there is no $j$ and afterwards that $j \notin T$ have passed to the ``while loop'' without return NO.

Therefore, if the algorithm returns $D$, we have that $d_{ii} \cdot a_{ij} = d_{jj} \cdot a_{ji}$ for any pairwise different $i, j$.

Suppose the algorithm returns NO when $A$ is symmetrizable.
Since $A$ is symmetrizable, the algorithm returns NO on line 17 and therefore there exist $i$ and $j$ such that $d_{ii} \cdot a_{ij} \neq d_{jj} \cdot a_{ji}$.

Let $T$ be the list $<i_1, \ldots, i_t>$ at the beginning of the ``while loop''. It follows from line 14, that first indices $\star$ of $T$ have $d_{\star\star} \neq 0$ and the last ones have $d_{\star\star}=0$, that is
there exists $k \in\{0, \ldots, t\}$ such that $d_{i_s i_s}\neq 0$ for all $1\leq s \leq k$ and
$d_{i_s i_s}= 0$ for all $k < s \leq t$. Observe that if $k=0$ that means that $d_{\star\star}=0$, for all $\star \in T$. Since $i$ is the first element in $T$ (line 5) and $d_{jj} \neq 0$ (line 15), we have that $d_{ii}\neq 0$ at the beginning of the ``while loop''. Therefore, $d_{ii}$ and $d_{jj}$ have been defined before and there exist $k \geq 3$, $i_1=i, i_2 =j$ and $i_3, \ldots, i_k \notin T$ such that
$ a_{i_1 i_2} \cdot a_{i_2 i_3} \cdot \ldots \cdot a_{i_k i_1} \neq 0$. Since $A$ is symmetrizable, we have by lemma 3.2, that $a_{i_1 i_2} \cdot a_{i_2 i_3} \cdot \ldots \cdot a_{i_k i_1} = a_{i_2 i_1} \cdot a_{i_3 i_2} \cdot \ldots \cdot a_{i_1 i_k}$. This implies that $d_{ii} \cdot a_{ij} = d_{jj} \cdot a_{ji}$, a contradiction.


\end{proof}

Next algorithm, do essentially the same as the above. The difference is that it not verify if $A$ is symmetrizable but assume it and end as soon as it have calculated all $d_{ii}$.
 
\begin{center}
\begin{algorithm2e}[H]
 \BlankLine
 \KwIn{A symmetrizable matrix $A$.}
 \KwOut{A diagonal matrix $D = d_{ii}$ of positive values such that $D \times A$ is symmetric.}
 \BlankLine

\ForEach{$i \in \{1, \ldots, n\}$}{ \label{sym-1}
	$d_{ii} \leftarrow 0$\label{sym-2}
}

$S \leftarrow \{1, \dots, n\}$ \\ \label{sym-3}
$T \leftarrow \{1, \dots, n\}$ ordered list

\While{$S \neq \varnothing$}{ \label{sym-5}
	 $i \leftarrow $ the first element of $T$ \\ \label{sym-6}
    	$T \leftarrow T \setminus \{i\}$\\ \label{sym-7}
 \If{$d_{ii} = 0$}{
		$d_{ii} \leftarrow 1$\\   
        $S \leftarrow S \setminus \{i\}$ }

	  \ForEach{$j \in T$} { \label{sym-9}
       				     \If{$a_{ji} \neq 0$} {
       				         move $j$ for the first position of $T$
        			       
				           \If{$d_{jj} = 0$}{
    			              $d_{jj} \leftarrow \frac{d_{ii} \cdot a_{ij}}{a_{ji}}$

              			     $S \leftarrow S \setminus \{j\}$
					         }
				      }
    }
}

\Return $D$ \label{sym-21}

\BlankLine

\BlankLine

\caption{$SymmetrizerMatrix(A)$ \label{alg:symmetrizer_matrix}}
\end{algorithm2e}
\end{center}

\begin{proposition}
Algorithm~$\ref{alg:symmetrizer_matrix}$ is correct.
\end{proposition}

\begin{proof}
The behavior of Algorithm~$\ref{alg:symmetrizer_matrix}$ is quite similar to the one of Algorithm~$\ref{alg:symmetrizable_matrix}$. The differences are that Algorithm~$\ref{alg:symmetrizer_matrix}$ does not check whether $A$ is symmetrizable or not and the use of an additional list $S$ which maintains the elements $i$ such that $d_{ii}$ is not defined yet. Since the algorithm assumes that the input matrix $A$ is symmetrizable, this allows to stop when all elements of $D$ have already been defined. This control is done by changing the condition of the ``while loop'' (line 5) accordingly and adding the operations to remove element $i$ of $S$ whenever $d_{ii}$ is defined (lines 10 and 16). 
\end{proof}

\section{Positive quasi-Cartan companion}
\label{sec:quasi_Cartan_companion}

A symmetric matrix $A$ is \textit{positive definite} if $x^T \cdot A \cdot x > 0$ for all vectors $x$ of length $n$, with $x \neq 0$. 
If the symmetrized matrix $D \times C$ is positive definite, then we say that the \textit{quasi-Cartan matrix $C$ is positive}. By Sylvester criterion~\cite{W1976}, be positive definite means that the principal minors of $D \times C$ are all positive.

The Theorem \ref{teo:equivalent_conditions} follows by Sylvester criterion and the fact that $det(C)$ is positive if and only if $det(D \times C)$ is.

\begin{theorem} \label{teo:equivalent_conditions}
Let $C$ be a symmetrizable matrix. The following conditions are equivalent:
\begin{enumerate}
	\item $C$ is positive.
	\item All principal minors of $C$ are positive.
	\item All leading principal minors of $C$ are positive.
\end{enumerate}
\end{theorem}

\begin{proposition}\label{teo:submatrix}
Let $B$ be a skew-symmetrizable matrix. $B$ has a positive quasi-Cartan companion matrix if and only if any principal submatrix of $B$ has a positive quasi-Cartan companion.
\end{proposition}

\begin{proof}
($\Rightarrow$)
Let $C$ be a positive quasi-Cartan companion of $B$. For induction, we just need to observe that $C_{[ii]}$ is symmetrizable matrix and therefore a positive quasi-Cartan companion of $B_{[ii]}$. Since $(D \times C)_{[ii]} = D_{[ii]} \times C_{[ii]}$ is symmetric, we have that $C_{[ii]}$ is a symmetrizable matrix. Similarly, $B_{[ii]}$ is skew-symmetrizable.

It follows from Theorem~\ref{teo:equivalent_conditions} that $C_{[ii]}$ is positive. Therefore, $C_{[ii]}$ is a positive quasi-Cartan companion of $B_{[ii]}$.

($\Leftarrow$)
Follows from the fact that the matrix $B$ is a principal submatrix of herself.
\end{proof}

We present the main ideas of the original proof of Lemma~\ref{lem_2.1}, since they are useful in the sequel.

\begin{lemma} (Barot, Geiss and Zelevinsky \cite{BGZ2006})\\  \label{lem_2.1}
Let $C$ be a positive quasi-Cartan matrix. Then
\begin{itemize}
	\item [(a)] $0 \leq c_{ij} \cdot c_{ji} \leq 3$ for any $i, j$ such that $i \neq j$.
	\item [(b)] $c_{ik} \cdot c_{kj} \cdot c_{ji} = c_{ki} \cdot c_{jk} \cdot c_{ij} \geq 0$ for any pairwise different $i, j, k$.
\end{itemize}
\end{lemma}

\begin{proof}
\begin{itemize}
	\item [(a)] Let $C' =  \left(
        \begin{array}{cc}
          2 & c_{ij} \\
          c_{ji} & 2 \\
        \end{array}
      \right)$ be a principal submatrix of $C$. Since $C$ is a symmetrizable matrix, it is symmetric by signs and $c_{ij} \cdot c_{ji} \geq 0$. Since $C$ is positive, then $det(C') = 4 - c_{ij} \cdot c_{ji} > 0$. Therefore, $c_{ij} \cdot c_{ji} \leq 3$.

      Since $C$ is symmetrizable, we have that sgn($c_{ij}$) = sgn($c_{ji}$). It follows that $c_{ij} \cdot c_{ji} \geq 0$.
	
	\item [(b)] Let $c_{ik} \cdot c_{kj} \cdot c_{ji} \neq 0$. Since that $C$ is symmetrizable, one can see that $c_{ki} \cdot c_{jk} \cdot c_{ij} = c_{ik} \cdot c_{kj} \cdot c_{ji}$. The condition of positivity to the principal minors $3 \times 3$ of $C$ in rows and columns $i, j, k$ can be rewritten as:
\begin{equation}
c_{ik} \cdot c_{kj} \cdot c_{ji} > c_{ij} \cdot c_{ji} + c_{ik} \cdot c_{ki} + c_{jk} \cdot c_{kj} - 4
\end{equation}	
Since $c_{ik} \cdot c_{kj} \cdot c_{ji} \neq 0$ we have that $|c_{st}| \geq 1$ and thus $c_{st} \cdot c_{ts} \geq 1$ for $(s,t) \in \{(i,k), (k,j), (j,i)\}$. Therefore, $c_{ik} \cdot c_{kj} \cdot c_{ji} > 3 - 4 = -1$. This yields the conclusion.
\end{itemize}
\end{proof}

We must observe that for a $3 \times 3$ quasi-Cartan matrix we have three leading principal submatrices that are: the $1 \times 1$ submatrix, that obviously has positive determinant; the $2 \times 2$ submatrix, that is positive due Lemma~\ref{lem_2.1}, when $0 \leq c_{ij} \cdot c_{ji} \leq 3$ and, finally, the $3 \times 3$ submatrix itself.

We define $C^+ = (c^+_{ij})$ such that $c^+_{ij} = |b_{ij}|$ and $c^+_{ii} = 2$.

\begin{proposition} \label{prop:all_positive}
Let $B$ be a skew-symmetrizable matrix of dimension $3 \times 3$. Then $B$ has a positive quasi-Cartan companion if and only if the matrix $C^+ = (c^+_{ij})$ is positive.
\end{proposition}

\begin{proof}
($\Rightarrow$) 
Suppose there exists a positive quasi-Cartan companion matrix $C$. Clearly, $0 \leq  c^+_{ij} \cdot  c^+_{ji} \leq 3$. Thus, $det(C^+) = 8 - 2 \cdot c^+_{jk} \cdot  c^+_{kj} - 2 \cdot c^+_{ij} \cdot  c^+_{ji} - 2 \cdot c^+_{ki} \cdot  c^+_{ik} + 2 \cdot c^+_{ij} \cdot  c^+_{jk} \cdot  c^+_{ki}$ $ = 8 - 2 \cdot c_{jk} \cdot  c_{kj} - 2 \cdot c_{ij} \cdot  c_{ji} - 2 \cdot c_{ki} \cdot  c_{ik} + 2 \cdot |c_{ij} \cdot  c_{jk} \cdot  c_{ki}| \geq det(C) > 0$. Since all leading principal minors of $C^+$ are positive, we have $C^+$ is also a positive quasi-Cartan companion of $B$.

($\Leftarrow$) 
Since $C$ is quasi-Cartan matrix, follows from Theorem~ \ref{prop:same_symmetrizer} that $C^+$ is a quasi-Cartan companion of $B$.
\end{proof}

\begin{lemma} \label{lemma_3x3}
Let $C$ be a $3 \times 3$ positive quasi-Cartan matrix.
\begin{enumerate}
	\item If $C$ is connected, then $0 \leq c_{ij} \cdot c_{ji} \leq 2$ for any $i, j$ such that $i \neq j$.
	\item $0 \leq c_{ik} \cdot c_{kj} \cdot c_{ji} \leq 2$ for any pairwise different $i, j, k$.
\end{enumerate}
\end{lemma}

\begin{proof}
The Proposition~\ref{prop:all_positive} shows that $det(C^+) > det(C)$. Thus, $C^+$ is positive. For sake of simplicity, consider that $C=C^+$.
\begin{enumerate}
	\item By Lemma~\ref{lem_2.1}, $0 \leq c_{ij} \cdot c_{ji} \leq 3$. Suppose, without loss of generality, that $c_{12}=3$. Then, $c_{21}=1$ and $det(C)=8 + 2\cdot c_{12} \cdot c_{23} \cdot c_{31} - 2 \cdot c_{12} \cdot c_{21} - 2 \cdot c_{13} \cdot c_{31} - 2 \cdot c_{23} \cdot c_{32} = 2 + 6 \cdot c_{23} \cdot c_{31} - 2 \cdot c_{13} \cdot c_{31} - 2 \cdot c_{23} \cdot c_{32} > 0$.
	
	Since $C$ is connected, $c_{23} \neq 0$ or $c_{31} \neq 0$. The positiveness of $C$ implies that both $c_{23}$ and $c_{31}$ are non zero. Recall that since $C$ is symmetrizable, we have that $c_{32} \neq 0$ and $c_{13} \neq 0$.
	
	Since $c_{12} \cdot c_{23} \cdot c_{31} = c_{21} \cdot c_{32} \cdot c_{13}$, we have that $3 \cdot c_{23} \cdot c_{31} = c_{32} \cdot c_{13}$ and $c_{32}=3$ or $c_{13}=3$. By symmetry, suppose that $c_{32}=3$. Then, $c_{23}=1$ and $c_{31}=c_{13}$. Since $c_{13}^2 = c_{13} \cdot c_{31} \leq 3$, we conclude that $c_{31}=c_{13}=1$.
	
	On the other way, $det(C)=2 + 6 \cdot c_{23} \cdot c_{31} - 2 \cdot c_{13} \cdot c_{31} - 2 \cdot c_{23} \cdot c_{32} = 2 + 6-2-6 = 0$. This yields a contradiction to the positiveness of $C$.
	
	Therefore, $c_{ij} \cdot c_{ji} \leq 2$, for all $i, j$.
	
	\item Suppose that $c_{12} \cdot c_{23} \cdot c_{31} = c_{13} \cdot c_{32} \cdot c_{21} \geq 3$. Suppose that $c_{12} \cdot c_{21}=3$, by the above, using the contrapositive of item 1, we have that $c_{23}=0$ and $c_{31} = 0$. A contradiction to the hypothesis. This implies that $c_{ij} \cdot c_{ji} \leq 2$ for all $i,j$ and that $c_{12} \cdot c_{23} \cdot c_{31} \geq 4$. We can suppose, without loss of generality, that $c_{12}=2$ and $c_{23}=2$. Thus, $c_{21}=1$ and $c_{32}=1$. Then, $c_{13} = 4$ and $c_{31} \neq 0$ is a contradiction to Lemma~\ref{lem_2.1}. Therefore, $0 \leq c_{12} \cdot c_{23} \cdot c_{31} \leq 2$.
\end{enumerate}
\end{proof}

\begin{theorem}\label{teo:connected}
Let $C$ be a positive quasi-Cartan matrix $n \times n$, with $n \geq 3$.
\begin{enumerate}
	\item If $C$ is connected, then $0 \leq c_{ij} \cdot c_{ji} \leq 2$ for all $i, j$.
	\item $0 \leq c_{ik} \cdot c_{kj} \cdot c_{ji} \leq 2$ for all pairwise different $i, j, k$.
\end{enumerate}
\end{theorem}

\begin{proof}
Suppose that exist $i, j$ such that $c_{ij} \cdot c_{ji} \geq 3$. Since $n \geq 3$, there exists $k \notin \{i,j\}$. Consider the principal submatrix $C'$ of $C$, formed of rows and columns $i, j$ and $k$.
By Lemma \ref{lemma_3x3}, $C'$ is disconnected. This implies that $c_{ik} = c_{kj} = 0$ and thus $c_{ki} = c_{jk} = 0$. Since this is true for all $k \notin \{i, j\}$, we have that $C$ is disconnected, a contradiction. The second item follows from Lemma \ref{lemma_3x3} by considering the principal submatrix of $C$, formed of rows and columns $i, j$ and $k$.
\end{proof}

\begin{theorem} \label{prop:all_positive2}
Let $B$ be a skew-symmetrizable matrix such that $b_{ij} \neq 0$ for all $i, j$. Then $B$ has a positive quasi-Cartan companion if and only if $C^+ = (c^+_{ij})$ defined by $c^+_{ij} = |b_{ij}|$ and $c^+_{ii} = 2$ is positive.
\end{theorem}

\begin{proof}
($\Rightarrow$) Suppose there exists a positive quasi-Cartan companion $C$. We will show by induction that $C^+$ is positive. If $C$ is $2 \times 2$ matrix, the result is clearly obtained. First, we show that $det(C^+) > 0$. Let $x_{ij} = sgn(c_{ij})$. Observe that $x_{ji} = x_{ij}$ and $x_{ii} =1$ due it to be a quasi-Cartan matrix. Since $c_{ij} \neq 0$ there is no ambiguity in this definition. Define $x_i = x_{1i}$ for all $i$. We will show, by induction on $i$, that $x_{ij} = x_i \cdot x_j$ for all $i, j$. Clearly, $x_{ii} = x_i^2$.

Since $x_1 = 1$, we clearly have that $x_{1i} = x_1 \cdot x_i$. Suppose that $x_{kj} = x_k \cdot x_j$ for all $k < i$. By Theorem~\ref{teo:connected}, $c_{ij} \cdot c_{jk} \cdot c_{ki} > 0$. Thus, $x_{ij} \cdot x_{jk} \cdot x_{ki} = 1$. Therefore, $x_{ij} = x_{jk} \cdot x_{ki} = x_{kj} \cdot x_{ki} = x_k \cdot x_j \cdot x_k \cdot x_i = x_i \cdot x_j$.

Since $x_{ij} = x_i \cdot x_j$ and $c_{ij}=x_{ij} \cdot c^+_{ij}$, we have that $C = X C^+ X$ where
$X =  \left(
        \begin{array}{ccc}
          x_1 & & 0 \\
          & \vdots & \\
          0 & & x_n \\
        \end{array}
      \right)$.

It follows that $det(C) = det(X) \cdot det(C^+) \cdot det(X) = det(C^+)$. By induction on dimension of $C^+$, we have that all leading principal minors of $C^+$ are positive.

($\Leftarrow$) 
It follows from the fact that $C$ is quasi-Cartan matrix and by Theorem~\ref{prop:same_symmetrizer}, that $C^+$ is a quasi-Cartan companion of $B$.
\end{proof}

We must observe that Theorem~\ref{prop:all_positive2} does not hold for all skew-symmetrizable matrices, as we can see in the following example.

\begin{example}
Let $B =  \left(
        \begin{array}{cccc}
          0 & 1 & 1 & 0 \\
          1 & 0 & 0 & 1 \\
          1 & 0 & 0 & 1 \\
          0 & 1 & 1 & 0 \\
        \end{array}
      \right)$,
      $C =  \left(
        \begin{array}{cccc}
          2 & -1 & 1 & 0 \\
          -1 & 2 & 0 & 1 \\
          1 & 0 & 2 & 1 \\
          0 & 1 & 1 & 2 \\
        \end{array}
      \right)$,
     $C^+ =  \left(
        \begin{array}{cccc}
          2 & 1 & 1 & 0 \\
          1 & 2 & 0 & 1 \\
          1 & 0 & 2 & 1 \\
          0 & 1 & 1 & 2 \\
        \end{array}
      \right)$. The quasi-Cartan companion $C^+$ is not positive, but $C$ is.
\end{example}

\begin{proposition}\label{prop:allone}
Let $B$ be the skew-symmetric matrix defined by $|b_{ij}|=1$ for all $i \neq j$, then $B$ has a positive quasi-Cartan companion.
\end{proposition}

\begin{proof}
%
If $B$ is a $n \times n$ matrix, then we can calculate that $det(C^+)=n+1$. Clearly, $C^+$ is positive. The result follows from Theorem~\ref{prop:all_positive2}.
\end{proof}

%

%


We present an algorithm with time complexity $\theta(n^4)$ to decide whether the given matrix $C$ is positive.

This algorithm is used as verifier and, thus, with the algorithm and the given matrix $C$, we prove that the problem of deciding if there exists a positive quasi-Cartan companion belongs to NP class. For more information about P, NP and NP-complete classes, see~\cite{CLSR2002,S2006}.

\begin{center}
\begin{minipage}{0.75 \textwidth}
\begin{algorithm2e}[H]
 \LinesNumbered
\small
 \BlankLine
 \KwIn{A symmetrizable $n \times n$ matrix $C$.}
 \KwOut{The response if the matrix is positive or not.}
 \BlankLine

 $n' \leftarrow n$\\
 $C' \leftarrow C$
 \BlankLine
	
\ForEach{$i \in \{1, \ldots, n\}$}{
	\If{($det(C') \leq 0$)}{
		\Return NO
	}
	$C' \leftarrow C'_{[n',n']}$\\
	$n' \leftarrow n'-1$\\
}

\Return YES

 \BlankLine

\BlankLine

\caption{$IsPositive(C)$ \label{alg:is_positive}}
\end{algorithm2e}
\end{minipage}
\end{center}

We also elaborate an exponential algorithm to find the positive quasi-Cartan companion of skew-symmetrizable matrix $B$.

\begin{center}
\begin{minipage}{0.85 \textwidth}
\begin{algorithm2e}[H]
\SetKwFunction{positive}{IsPositive}
 \LinesNumbered
\small
 \BlankLine
 \KwIn{A skew-symmetrizable matrix $B$.}
 \KwOut{A positive quasi-Cartan companion $C$, if there exists.}
 \BlankLine

 $C \leftarrow |B|$ \tcc{\footnotesize The matrix $C$ is initialized with $B$ by positive entries.}

 \ForEach{$i \in \{1, \ldots, n\}$}{
 	$c_{ii} \leftarrow 2$\\
 }

 \eIf{\positive{C}}{
 	\Return $C$
 }
 {
 	\ForEach{$x \in \ \{(x_{ij}) | x_{ij} \in \{-1, 1\}$ and $i < j\}$}{
	
		\ForEach{$i \in \{1, \ldots, n\}$}{
			\ForEach{$j \in \{i+1, \ldots, n\}$}{
				$c_{ij} \leftarrow x_{ij} \cdot |b_{ij}|$\\
				$c_{ji} \leftarrow x_{ij} \cdot |b_{ji}|$
			}
		}
		{
			\If{\positive{C}}{
				\Return $C$
			}
		}
	}
}

\Return ``There is no positive quasi-Cartan companion of $B$''

 \BlankLine

\BlankLine

\caption{$PositiveQuasiCartanCompanion(B)$ \label{alg:quasi_Cartan}}
\end{algorithm2e}
\end{minipage}
\end{center}

\section{Conclusions}
\label{sec:conclusions}

In this paper, we present some mathematical properties of symmetrizable, skew-symmetrizable and positive quasi-Cartan matrices. These matrices are important in the context of cluster algebra, to decide if it is of finite type.

Importantly, the skew-symmetric matrices represent directed graphs that arise from the cluster algebras, called quivers.

We also developed two polynomial algorithms for symmetrizers matrices: one to decide whether a symmetrizer matrix exists with time complexity $\theta(n^2)$, and another to find a symmetrizer matrix, if there exists, with time complexity $\theta(n^2)$ in the worst case and $\theta(n)$ in the best case.

For a skew-symmetrizable matrix, we present an algorithm to decide whether a matrix is positive or not and another, exponential, to find a positive quasi-Cartan companion, if there exists.

\section*{Acknowledgments}
\label{sec:acknowledgments}

We would like to thank the Funda\c{c}\~{a}o de Amparo \`{a} Pesquisa do Estado de Goi\'{a}s (FAPEG) and the Coordena\c{c}\~{a}o de Aperfei\c{c}oamento de Pessoal de N\'{i}vel Superior (CAPES) for partial support given to this research by Programa Nacional de Coopera\c{c}\~{a}o Acad\^{e}mica (Procad) and also the Conselho Nacional de Desenvolvimento Cient\'{i}fico e Tecnol\'{o}gico (CNPq) and Funda\c{c}\~{a}o Carlos Chagas Filho de Amparo \`{a} Pesquisa do Estado do Rio de Janeiro (FAPERJ).


\bibliographystyle{apalike}

\end{document}